\documentclass[sigconf]{acmart}
\setcopyright{acmcopyright}
\copyrightyear{2021}
\acmYear{2021}

\usepackage{xcolor}
\usepackage{graphicx}
\usepackage{amsfonts}
\usepackage{amsmath,amsthm}
\usepackage{enumitem}
\usepackage{mathtools}
\usepackage{booktabs}
\usepackage{caption}
\usepackage{subcaption}
\usepackage{algorithm}

\acmYear{2021}\copyrightyear{2021}
\acmConference[AFT '21]{3rd ACM Conference on Advances in Financial Technologies}{September 26--28, 2021}{Arlington, VA, USA}
\acmBooktitle{3rd ACM Conference on Advances in Financial Technologies (AFT '21), September 26--28, 2021, Arlington, VA, USA}
\acmPrice{15.00}
\acmDOI{10.1145/3479722.3480996}
\acmISBN{978-1-4503-9082-8/21/09}
\begin{document}
\pagestyle{plain}
\title{Private Attacks in Longest Chain Proof-of-stake Protocols with Single Secret Leader Elections}

\author{Sarah Azouvi}
\affiliation{%
Protocol Labs
}
\author{Daniele Cappelletti}
\affiliation{%
Politecnico di Torino
}

\newcommand{\SSLE}[2]{P_{\text{SSLE}}^#1(#2)}
\newcommand{\PLE}[2]{P_{\text{PLE}}^#1(#2)}
\newcommand{\Bin}{\textsf{Bin}}
\newcommand{\var}{\textsf{Var}}
\newcommand{\E}{\mathbb{E}}
\newcommand{\gap}{\mathcal{G}}
           
\newcommand\sarah[1]{\textcolor{violet}{Sarah: #1}}
\begin{abstract}
Single Secret Leader Elections have recently been proposed
as an improved leader election mechanism for proof-of-stake (PoS) blockchains. However, the security gain they provide has not been quantified.
In this work, we present a comparison of PoS longest-chain protocols
that are based on Single Secret Leader Elections (SSLE) -- that elect exactly one leader per round -- versus those based on Probabilistic Leader Elections (PLE) -- where one leader is elected on expectation.
Our analysis shows that when considering the private attack -- the worst attack on longest-chain protocols~\cite{dembo2020everything} -- the security gained from using SSLE is substantial: the settlement time is decreased by $\sim25\%$ for a 33\% or 25\% adversary.
Furthermore, when considering grinding attacks, we find that the security threshold is increased by 10\% (from 0.26 in the PLE case to 0.36 in the SSLE case) and the settlement time is decreased by roughly $70\%$ for a 20\% adversary in the SSLE case.
\end{abstract}

\maketitle   
\section{Introduction}
Proof-of-stake has been proposed as a more energy-efficient alternative consensus protocol to proof-of-work for cryptocurrencies.
In proof-of-work, miners need to solve a computational puzzle in order to earn the right to create a block and receive the associated financial rewards.
The amount of blocks that they mine is, hence, proportional to their computational power.
In contrast, the idea behind PoS is that participants mine a fraction of blocks that is proportional to the relative amount of coins they own.
One crucial component of PoS consensus protocols is their \emph{leader election}~\cite{caucus}, used to decide which participants will get to create the next block.
Although leader election protocols have been studied widely in the traditional field of distributed systems, 
the setting of blockchains -- which are decentralized and provide financial rewards for block creation -- poses new challenges that are paramount to the security of the whole consensus protocol.
For example, leader election protocols must be fair, in the sense that miners must be elected proportionally to their power (stake or compute resources); this is to ensure they are compensated fairly for their investment and to prevent Sybil attacks.
They should also be private, i.e., no actor should be able to guess the next leader until they broadcast their block with a proof-of-eligibility; this prevents denial-of-service (DoS) attacks against the next leader.
Many attempts have been made to design an adequate leader election protocol, such as using hash functions~\cite{daian2019snow,caucus} or coin tossing protocols~\cite{kiayias2017ouroboros}.
Another method that has been adopted by many protocols~\cite{algorand,praos,nakamoto-pos,praos} is the use of verifiable random functions (VRFs)~\cite{micali1999verifiable}.
%VRF produce different pseudo-random number for each 

Most leader election protocols used in blockchains are private but probabilistic~\cite{algorand,praos,daian2019snow}, meaning that one leader will
privately be elected on expectation, but it could be that zero or several leaders are elected in some rounds. 
Probabilistic Leader Elections (PLEs) can be problematic since having multiple leaders elected in a round leads to different views or \emph{forks} in the system.
Ouroboros~\cite{kiayias2017ouroboros} proposed a leader election
where exactly one leader is elected per round. However, the election is not private and exposes the leader to DoS attacks.

To solve these problems, Single Secret Leader Election (SSLE), where \emph{exactly} one leader is privately elected at each round, has been proposed in~\cite{ssle,catalanoefficient}.
Although current SSLEs are more complex to implement than, for example, PLEs based on VRFs, they intuitively improve the security of PoS blockchains because they reduce the probability of honest forks.
However, there is no formal proof of this statement and, if true, the exact gain in security that they achieve, compared to PLEs, has yet to be quantified.

\paragraph{Our contribution}
In this work, we perform a comparison between SSLEs and PLEs and investigate the gain in security against private attacks, where the adversary grows a private chain of blocks so as to outpace the honest chain.
Focusing our analysis on private attacks, and not general adversaries, is motivated by the work of Dembo et al.~\cite{dembo2020everything}, who showed that private attacks are the worst attack in longest-chain blockchains, in the sense that the true security threshold of a longest-chain protocol is the same as the security threshold against private attacks.

We start our analysis by considering leader elections that have access to a perfect source of randomness in Section~\ref{sec:analysis}, before considering randomness
that is derived from the blockchain itself -- as is the case for many
PoS longest-chain protocols~\cite{praos,nakamoto-pos,praos} --
and the resulting grinding attack in Section~\ref{sec:grinding}.
We find that, with perfect randomness, the persistence parameter (or settlement time) against private attacks in a synchronous network is decreased by roughly $25\%$ against a 33\% or 25\% adversary. In the grinding case, the security threshold is higher by 10 percentage points in the SSLE case (36\%) than
in the PLE case (26\%) and the persistence parameter is decreased by roughly 70\%.
Although it is not surprising that SSLEs perform better than PLEs in longest-chain blockchains, we did not expect to see such significant improvements.

%For this reason, we focus, in this paper, on this particular attack.
%We make some additional simplifying assumptions: that the network is synchronous and the source of randomness used perfectly secure.
These results are very encouraging and
could motivate the switch from PLE
to SSLE in current PoS blockchains.
Since SSLEs are still less efficient than PLEs,  quantifying the gain in security is paramount for evaluating fairly the trade-offs between the two.
We leave a full analysis of PoS with SSLE in a partially synchronous network and against all possible attacks as future work. Whether with SSLE or PLE, obtaining an exact formula for the security of the protocol against any adversary (i.e., the probability of breaking the consensus) is an open and non-trivial problem. All the known analyses that consider a general adversary are based on bounds that would not be useful for our comparison, and hence a general adversary is beyond the scope of this work.
%Due to the simplifying assumptions we made, our work is the first to give an exact formula for the security of longest-chain PoS protocols under private attacks.

Lastly, we note that our analysis does not make use of the secret property of the SSLE (that protects against DoS attacks), and thus the results would hold for a public single leader election as well.

\section{Related work}\label{sec:related}
Previous works have proposed formal analyses of PoS blockchains based on PLEs~\cite{linearconsistency,nakamoto-pos,posat}.
In most of them, the proof revolves around
finding a \emph{special} block that guarantees
security of the protocol and bounding the probability that this type of block does not appear in a sequence of $n$ consecutive blocks.
For example, Kiayias et al.~\cite{kiayias2020consistency} consider \emph{Catalan} blocks, Bagaria et al.~\cite{nakamoto-pos} and Deb et al.~\cite{posat} consider Nakamoto blocks, Daian et al.~\cite{daian2019snow} consider \emph{pivot} point.
In~\cite{nakamoto-pos}, the interpretation given for that special block
is that no private chain started by the adversary at any point in time before that block will be able to catch up with any chain that includes that block at any point in the future. They show that any attack can be modelled as a composition of private attacks.

Dembo et al.~\cite{dembo2020everything}
considered three different types of longest-chain blockchain protocols: proof-of-work, proof-of-stake, and proof-of-space, and showed that in all three cases the true security threshold of the protocol (i.e., the security threshold when considering all attacks) is exactly the same as the security threshold of the private attack.
They also prove that the private attack is the worst attack in the synchronous proof-of-work case.
In the case of PoS protocols, based on the fact that other attacks do not appear in the security threshold, they conjecture that these attacks should have at least the same exponent as the private attack.
Even though this conjecture is left open, we decide in this paper to focus on the private attack and leave a full analysis as future work.
While this limits the generality of our work, we believe that, due to the conjecture above, studying the private attack alone gives a good proxy for the security of PoS longest-chain protocols.
Furthermore, this simplifying assumption allows us to get exact bounds, as opposed to the loose bounds obtained by other analyses~\cite{daian2019snow,nakamoto-pos,posat,linearconsistency}, usually based on Chebyshev's inequality or similar inequalities.
% Talk about the papers that do formal proof on longest-chain protocols such:
% Ouroboros Praos~\cite{praos}
% Linear consistency~\cite{linearconsistency} etc
% Nakamoto PoS~\cite{nakamoto-pos}
% PoSAT~\cite{posat}.
% ~\cite{linearconsistency} consider a model somehow similar to SSLE as the consider honest vs adversarial slots but their protocol is based on Probabilistic
%Leader election. + they consider big bounds unlike our work that has a different purpose to compare the two protocols.  
Loose bounds on the probability of successful attacks would not be sufficient to meaningfully quantify the difference between SSLE and PLE in longest-chain PoS blockchains as even if one bound is smaller than the other, this does not say anything about how the exact probabilities compare.
% able us to capture the gain in security in using SSLE vs PLE as one case could have a smaller bound than the other but be less secure.
Simplifying our model to only consider private attacks alleviates this problem as it allows us to get closed-form solutions that we are meaningfully able to compare.
% One work by Li et al.~\cite{li2020close} does give an exact formula for the probability that an adversary violates persistence, however in the context of proof-of-work and such formula remains an open problem for PoS blockchains.

\section{Background, Model, and Definitions}\label{sec:model}
% \sarah{
% Potential next step:
% consider balance attacks,
% consider multiple leaders per round (e.g. EC)}
We start by giving some background on PoS blockchains before presenting our model.
Because we limit our adversary to private attacks, 
%motivated by the work of Dembo et al.~\cite{dembo2020everything}
we will consider a rather simplistic model and abstract away most of the blockchain concepts.

\subsection{Proof-of-stake blockchains}
A blockchain is a digital distributed ledger of transactions. Transactions are grouped into \emph{blocks} that are chronologically ordered in a linked chain.
In a PoS system, participants -- also called miners or validators--
are eligible to create blocks based on their relative stake, as measured from the number of coins that they own.
In longest-chain protocols, eligible miners create their block and append it to the longest chain of blocks that they are aware of, i.e., the new block should link to the block atop the longest existing chain of blocks.

Briefly, the protocol works as follows. Whenever miners receive a block, they add it to their local view.
At each round, they run a leader election protocol that randomly decides their eligibility for that round, proportionally to their stake. 
For now we assume access to a perfect \emph{random beacon}~\cite{ben1985collective} that emits a random number at the beginning of each round. This beacon is then given as an input to the leader election. 
If they are elected in that round, miners create a block on top of their longest chain and broadcast their block to the network.
In the case where there exist two chains of the same length, e.g., if more than one leader was elected in the previous round, miners break the tie in a random way.
In our model, we will ignore the content of the blocks as well as associated reward or financial incentives.
% \begin{algorithm}[H]
% \caption{Proof-of-stake consensus}
% \label{alg:generator}

% Map store=new Map(obj, queue)\;
% \generate{Object pivot}{
%      \ForAll{child $c$ in pivot}{
%      \If{ $c$'s FieldContext is not set and $c$ is fusible}{
%           generate($c$)\;
%       }
%      }
%      build pivot's fieldContext $fc$\;
%      EmitClassName\;
%      EmitFields($fc$)\;
%      EmitMethods($fc$)\;
% }
% \end{algorithm}

Informally, a blockchain should verify two security properties~\cite{dembo2020everything}:
(1) \emph{persistence}, meaning that 
once a block has been confirmed by an honest miner, it should stay in the chain of that miner indefinitely i.e., it is not possible for an adversary to revert the chain of an honest miner, except 
for the last $k$ blocks; and 
(2) \emph{liveness}, meaning that new blocks should be appended to the chain continually, even in the presence of an adversary.
% (2) Chain quality meaning that an adversary should contribute a limited fraction of the total blocks in the blockchain.

One important parameter in longest-chain blockchains is the 
\emph{persistence parameter} -- or settlement time -- $k$, which, informally, represents the number of rounds after which an honest miner will consider a block confirmed, i.e., after which it will stay in the chain forever.

% \subsubsection{Random Beacon}
% These were first introduced by Rabin [29] as a service for “emitting at regularly spaced time intervals, randomly chosen integers”.

\subsection{Private attacks}
Private attacks are a specific type of attack on
longest-chain blockchains, in which
an adversary keeps its blocks private (instead of sending them to the rest of the miners) and does not mine on other participants' blocks. In other words, the adversary is creating its own chain, parallel to the honest chain.
The adversary succeeds if it can create a private chain longer than the honest chain. If that is the case, it will broadcast its chain to the rest of the players, forcing them to abandon their own chain and mine on top of the adversarial chain.
It is clear that this attack will almost surely be successful if the adversary can mine blocks at a rate faster than the rest of the players.
In PoS blockchains, this rate is proportional to the relative amount of stake that one owns. Hence this attack will succeed if the adversary owns more than half of the total stake.
If not, then there exists a time after which
this attack becomes very unlikely to succeed. 
% We are interested in computing this specific length and the associated probability of success.

\subsection{Model}

\subsubsection{Assumptions}
We consider a set of $N$ participants, a fraction $\alpha$
of which are controlled by an adversary $\mathcal{A}$ performing a private attack as specified above. The remaining participants are 
\emph{honest} and follow the protocol.
Time is divided into discrete time-steps.
We assume that we have access to a broadcast algorithm and we consider a synchronous model meaning that each message, i.e., block, sent by an honest participant reaches everyone after
at most $\Delta$ time steps with $\Delta>0$.
We consider a round-based protocol that relies on an underlying leader election mechanism as defined below.
We assume that the duration of each round is strictly more than
$\Delta$, such that any message sent by any honest participant at the beginning of a round will be received before the end of that round.
We assume that we have access to a perfect random beacon that emits new randomness at the beginning of each round. We will relax this assumption in Section~\ref{sec:grinding}.

% \begin{definition}[Leader Election Protocol]
% $S\gets\mathcal{LE}(n,r)$

% \end{definition}
For simplicity, we consider a flat model, meaning that each miner accounts for one unit of stake and that the set of participants is static.

\subsubsection{Leader Election}

Every PoS protocol uses a leader election to decide which miners are eligible to create blocks.
%Leader election is in itself a difficult problem to solve. 
In this paper, we treat it as a black box algorithm that takes as input a random number
$r$ and a set of participants and outputs a (potentially empty) set of leaders.
We consider two types of leader election:
Single Secret Leader Election (SSLE)~\cite{ssle}, where \emph{exactly} one leader is elected per round, and Probabilistic Leader Election (PLE)~\cite{algorand}, where one leader is elected per round \emph{on expectation}. This means that, when using a PLE, there could be multiple leaders or no leader at all in a round.
We denote $a_n$ and $h_n$ the number of adversarial and honest leaders elected at round $n$, respectively.
% Following~\cite{nakamoto-pos}, in the PLE case, we model the number of elected leader as a Poisson process.

 We model the PLE case as follows: at each round, every player ``tosses their own coin'' (using the randomness given by the random beacon) to determine their eligibility; each player wins with probability $1/N$. 
In practice, this is achieved using a verifiable random function~\cite{micali1999verifiable}: each player uses the VRF to compute their own random number; if it falls below a threshold, they are elected leader. 
% Each player computes their own VRF that produce a random number and if it falls below a threshold, they are elected leader.
An adversary controlling a fraction $\alpha$ of the players will thus have a number of leaders that follows a Binomial distribution with $N\times\alpha$ trials and success probability $1/N$, seeing as they get to toss a coin for each of the players they control.
We assume that the number of participants $N$ is big and hence the number of adversarial leaders
elected can be approximated as a
Poisson distribution with parameter $\alpha$.
Similarly, the number of honest participants elected in a round follows a Poisson distribution with parameter $1-\alpha$.
Furthermore, the number of adversarial and honest leaders are independent from each other.

In the SSLE case, exactly one of the players is elected, hence the number of adversarial leaders follows a Bernoulli distribution with parameter $\alpha$. However, the number of adversarial and honest leaders are not independent. In particular, the number of honest leaders is the complement of the number of adversarial leaders, i.e., $h_n = 1-a_n$ for every $n\in\mathbb{N^*}$.

% Many PoS protocols require that the leader election protocol is
% \emph{private} meaning that participants can only see for themselves whether they are elected or not but they cannot predict which other participants will be elected before they reveal their proof of eligibility. 
In this paper, where we consider a static adversary that performs a private attack, we ignore some
of the practical requirements for leader elections in PoS blockchains (such as unpredictability and secrecy) that are not
relevant to the model.

\subsubsection{Security games}
We consider an adversary mounting a private attack against the honest players.
%The adversary succeeds in a private attack of length $n$ if it can create a private chain of length $n$ that is longer or equal than the honest chain.
%The honest miners and the adversary hence mine on separate chains for the duration of the attack.
Accordingly, we define the following games that capture whether or not the adversary succeeds in a private attack according to the assumptions above.

\begin{definition}[$(L,\alpha)$-PLE Private Game]
The PLE private game with parameters $(L,\alpha)$ is defined as follows:
at each round $n\in [1,\dots,L ]$
a number $a_n$ of adversarial leaders and
$h_n$  of honest leaders are selected at random from, respectively,
Poisson distributions of parameters $\alpha$ and $1-\alpha$.
We say that the adversary wins the PLE private game of length $L$ and power $\alpha$ if the number of rounds with non-zero adversarial leaders is greater than or equal to the number of rounds with non-zero honest leaders, i.e.:
$$|\{n \in [1,\dots,L ]: a_n>0\}|\ge |\{n \in [1,\dots,L ]: h_n>0\}|.$$
% and we note $\PLE{L}{\alpha}$ the associated probability.

\end{definition}

\begin{definition}[$(L,\alpha)$-SSLE Private Game]
The SSLE private game  with parameters $(L,\alpha)$ is defined as follows:
at each round $n\in [1,\dots,L ]$,
exactly one leader is elected. This leader is adversarial ($a_n=1,h_n=0$) with probability $\alpha$ and honest ($a_n=0,h_n=1$) with probability $1-\alpha$.
We say that the adversary wins the SSLE private game of length $L$ and power $\alpha$ if the number of rounds with  adversarial leaders is greater than or equal to the number of rounds with honest leaders, i.e.:
$$|\{n \in [1,\dots,L ]: a_n=1\}|\ge |\{n \in [1,\dots,L ]: h_n=1\}|.$$
%and we note $\SSLE{l}{\alpha}$ the associated probability.
%We note $\SSLE{\ge l}{\alpha}$ the probability that the adversary wins the SSLE game of length at least $l$, i.e., there exists a round $l_0'\ge l$ where the adversary wins the game.
\end{definition}

% In some cases the analysis for the SSLE and PLE private games will be the same and hence we proceed by doing a general analysis of what we call the \emph{private game} or \emph{private attack}.
% We commonly refer to the SSLE and PLE private games as simply the private games.
% At time we will talk about the private game 

We now define the persistence parameter $n_0$ of the games, parametric in $0<\epsilon<1$, that intuitively represents the number of rounds after which the adversary cannot win the private game, except with probability $\epsilon$.

\begin{definition}[$\epsilon-$persistence parameter]
We say that $n_0$ is the $\epsilon-$persistence parameter of the SSLE, resp. PLE, private game if the probability that there exists any $n\ge n_0$ such that the adversary wins the SSLE, resp. PLE, game of length $n$ is $\epsilon$.

\end{definition}

In order to study the private games, we define the concept of gap, already introduced by Blum et al.~\cite{linearconsistency}.

\begin{definition}[Gap]
The gap at round $n\in[1,\dots,L]$ is the difference between the number of adversarial rounds and honest rounds in rounds 1 to $n$.
Let
$\gap_n^{SSLE}(\alpha)$ and $\gap_n^{PLE}(\alpha)$ denote, respectively, the gap in the PLE and SSLE private games of parameters $(L,\alpha)$. For $n\in[1,\dots,L]$, we have:
% \[\gap = \# \text{adversarial rounds } - \# \text{honest rounds }\]
\[\gap_n^{SSLE}(\alpha) = |\{i \in [1,\dots,n ]: a_i=1\}|- |\{n \in [1,\dots,n ]: h_i=1\}|
\]
\[\gap_n^{PLE}(\alpha) = |\{i \in [1,\dots,n ]: a_i>0\}|- |\{n \in [1,\dots,n ]: h_i>0\}|
\]
\end{definition}
If we consider a general analysis that applies to both settings (SSLE and PLE), we simply write $\gap_n$ and talk about the private game.
It is clear that the adversary wins the
private game of length $n$ if and only if $\gap_n\ge 0$.
% We are interested in finding 
% $n_0$ the probability that the adversary wins the SSLE, resp. PLE, game of length greater or equal than $n_0$.
% Intuitively, in the context of blockchains $n_0$ can be considered as a finality parameters as this means that the adversary cannot reverse any chain of length greater than $n_0$.
In the next section, we will study the behaviour of the gap. We are specifically interested in the probability that the adversary wins the PLE and SSLE games for any $n\ge n_0$,  $n_0\in\mathbb{N}$.

Before this, we briefly explain why the PLE and SSLE private games are an accurate description of the private attack in PoS systems.
In the SSLE case, since there is exactly one leader per round and the network is synchronous, it is clear that the honest chain will
be exactly the same length as the number of honest rounds and the adversarial chain will be at most the same length as the number of adversarial rounds.

In the PLE case, however, there could be honest forks due to multiple honest leaders being elected in the same round. 
If that happens, and since we consider a synchronous network, the longest chain will still increase by one even if some of the blocks at that round are being abandoned. Even in the worst case where there are multiple longest chains for several rounds, each of them will still be as long as the number of honest rounds.
Similarly, in the adversarial case, even if the adversary has more than one block on one round, it can only append one block per round and, hence, its longest chain is bounded by the number of eligible rounds.

\section{Analysis}\label{sec:analysis}
% We proceed as follows: we first remark that the gap is a biased random walk and specify, in each case, with what probabilities. We then study the expectation and variance of the random walks before stating the main theorems that give a closed-form solution for the probability that an adversary succeeds in winning the game for any length greater than or equal than $n_0$.
In this section we prove our main theorems, Theorem~\ref{thm:ssle} and~\ref{thm:ple}, where we express for $n_0\in\mathbb{N}^*$ the probability that an adversary succeeds in winning the private game for any length greater than or equal to $n_0$. 
%This probability gives us access to the persistence parameter: if we note $\epsilon$ this probability for $n_0\in\mathbb{N}$, then $n_0$ is the $\epsilon$-persistence parameter.
This probability corresponds to the value $\epsilon$ for the corresponding $\epsilon-$persistence parameter $n_0$. 
We will then compare the $\epsilon-$persistence parameter in the SSLE and PLE cases.

\paragraph{SSLE and PLE games as biased random walks}
In the SSLE game, it is straightforward to see that the gap will increase by one with probability $\alpha$ and decrease by one with probability $1-\alpha$.

In the PLE game, there are two events in which the gap will not change.
The first event is when
no leader is elected. The second event is when
an honest leader is elected 
at the same round as an adversarial leader. 
%Hence in the PLE game the two following events:
%``no leader is elected in that round'' and ``both adversarial and honest leaders are elected in that round''
%do not change the gap. 
We call these events \emph{null events} and denote $p_0$ the probability that they happen.

On the other hand having multiple adversarial leaders and no honest leader is equivalent to having exactly one adversarial leader since the gap will grow by one at that round regardless of the exact number of adversarial leaders and vice versa for honest leaders. We note $p_a= \Pr[a_n>0\text{ and }h_n=0]$ the probability that an adversary is the unique leader in a round - which does not depend on $n$ - and $p_h=\Pr[a_n=0\text{ and }h_n>0]$ the probability that the honest players are unique leaders.
The PLE game can be modeled as a random walk that increases by one with probability $p_a$, decreases by one with probability $p_h$ and stays the same with probability $p_0=1-p_a-p_h$.
Since $a_n$ and $h_n$ are independent, we have:
\begin{align*}
   p_a &=  \Pr[h_n =0 ]\times \Pr[a_n\ge 1] \\
   % &=e^{-\lambda_h}(1-e^{-\lambda_a})\\
    &=e^{\alpha-1}(1-e^{-\alpha})\\
    &= e^{\alpha-1}-e^{-1}
\end{align*}
And similarly:
$p_h=e^{-\alpha}-e^{-1}$.

We now move on to prove our first lemma.
For the rest of the paper, for $p\in (0,1)$ and $n\in\mathbb{N}$, we note 
%$\Bin(p,n,k)=\binom{n}{k}p^k(1-p)^{n-k}$ for $k\in\mathbb{N}$ and $\Bin(p,n,k)= 0$ for $k\in\mathbb{R}\setminus\mathbb{N}$.

\[
\Bin(p,n,k) = \left\{\begin{array}{ll}
       \binom{n}{k}p^k(1-p)^{n-k}, &  \text{ for } k\in\mathbb{N}\; \\
        0, & \text{for } k\in\mathbb{R}\setminus\mathbb{N} 
        \end{array}\right.
\]

\begin{lemma} For every $(n,v)\in\mathbb{N}^2$ and $\alpha\in \; (0,1)$:
$$\Pr[\gap_n^{SSLE}(\alpha)=v] = \Bin(\alpha,n,\frac{1}{2}(n+v)). $$
% \[
% \Pr[\gap_n^{SSLE}(\alpha)=v] = \left\{\begin{array}{ll}
%         \Bin(\alpha,n,\frac{1}{2}(n+v)), &  \text{ for } n+v\equiv 0\; [2]\text{ and }\\ & -n\le  v\le n\\\
%         0, & \text{otherwise } 
%         \end{array}\right.
% \]
\end{lemma}
\begin{proof}
We already noted that 
$(\gap_n^{SSLE}(\alpha))_{n\in\mathbb{N}}$ is a random walk such that 
\[\gap_{n+1}^{SSLE}(\alpha)= \left\{\begin{array}{ll}
        \gap_n^{SSLE}(\alpha)+1, & \text{ with probability }\alpha\\
        \gap_n^{SSLE}(\alpha)-1, & \text{ with probability }1-\alpha
        \end{array}\right.
\]
The proof of the lemma follows from standard results on random walks and can be found in~\cite{proba-random-walk}.
Briefly, 
for $n\in\mathbb{N}$:  $-n\le\gap_n^{SSLE}(\alpha)\le n$.
We note $u$ the number of times that $\gap_n^{SSLE}(\alpha)$ increased by one and $d$ the number of times that $\gap_n^{SSLE}(\alpha)$ decreased by one. 
If $\gap_n^{SSLE}(\alpha)=v$ for $v\in [-n,n]$,
we have $u-d = v$ and $u+d = n$ hence $u=\frac{1}{2}(v+n)$. There are exactly $\binom{n}{u}$
different ways to reach $v$, starting from $0$ and hence
$\Pr[\gap_n^{SSLE}(\alpha)=v]=\binom{n}{u}\alpha^u(1-\alpha)^{n-u}$.
\end{proof}
We now look at the equivalent lemma, in the PLE case.
\begin{lemma} For every $(n,v)\in\mathbb{N}^2$ and $\alpha\in \; (0,1)$:
\[\Pr[\gap^{PLE}_n(\alpha)=v]=\sum_{l=0}^{n-v} \Bin(p_0,n,l)\Pr[\gap^{SSLE}_{n-l}(\frac{p_a}{1-p_0})=v]
\]
\end{lemma}

\begin{proof}

If $\gap_n^{PLE}(\alpha) = v$, there can be between 0 and $n-v$ null slots (i.e., slots where the gap does not change from the previous step); hence, we have:
\begin{align*}
\Pr[\gap_n^{PLE}(\alpha)=v]&=\sum_{l=0}^{n-v}\Pr[\gap_n^{PLE}(\alpha)=v|\text{\# null events} = l]
\times \\
&\Pr[\text{\# null events} = l]\end{align*}

We start by assuming that there exist exactly $l$ null events in the PLE game.
% Then the probability that $\gap^{PLE}_n(\alpha)$ is equal to $v$ is the same as the probability that
% $\gap^{SSLE}_{n-l}(p'_a)$ is equal to $v$
% where $p'_a$ is the adjusted probability of winning of the adversary after null events have been removed.
After removing the $l$ null events from the PLE game, the remaining $n-l$ slots are either fully adversarial or fully  honest; this is equivalent to an SSLE game of length $n-l$.
The power of the adversary in this new SSLE game needs to be adjusted, accounting for the fact that null events have been removed.
Hence the equivalent SSLE game has parameters $n-l$ and $\frac{p_a}{1-p_0}$
due to the independence of each round. 
 We thus have the following:

\begin{align*}
    \Pr[\gap_n^{PLE}(\alpha)=v]=\sum_{l=0}^{n-v} \Pr[\gap^{SSLE}_{n-l}({\frac{p_a}{1-p_0}})=v]\times\Pr[\text{\# null events} = l]
\end{align*}
Following standard results on Binomial distribution we have that $\Pr[\text{\# null events} = l]=\Bin(p_0,n-v,l)$, hence:
\begin{align*}
 \Pr[\gap_n^{PLE}(\alpha)=v]=\sum_{l=0}^{n-v} \Pr[\gap^{SSLE}_{n-l}({\frac{p_a}{1-p_0}})=v]\times
    \Bin(p_0,n-v,l).
\end{align*}
This result can also be derived using the multinomial distribution.
\end{proof}

Having computed the probability that the gap is equal to some value $v$, and in order to study the persistence parameter, we are also interested in the probability that the gap, starting at some negative value $-M$, goes back up to zero.
In blockchain terms, if the adversarial chain is behind the honest chain by $M$ blocks, we compute the probability that it eventually catches back.

\begin{lemma}
For $\alpha<1/2$, if in round $n_0\in\mathbb{N}$, $\gap_{n_0} =- M$ for $M>0$, then the probabilities $r_M$ that $\gap_n$ ever reaches $0$ for any $n\ge n_0$ in the SSLE and PLE cases are:
\[
r_M^{SSLE} =\left( \frac{\alpha}{1-\alpha}\right)^M;\; r_M^{PLE} =\left( \frac{e^{\alpha}-1}{e^{1-\alpha}-1}\right)^M
\]
\end{lemma}

\begin{proof}
% For a biased random walk, let note $p$ and $q$ the probability that the walk increases (resp. decreases) by one.
We start by considering 
the PLE case.
The probability $r_M$ that $\gap_n^{PLE}$ reaches 0 starting from a position $-M$ for $M>0$ is the same as the
probability that $\gap_n^{PLE}$ ever reaches M when starting at 0 (i.e., $\gap_n^{PLE}$ has a net increase of M). 

Let's note $r_1=r$. Then we have $r_M=r^M$ ($\gap_n^{PLE}$ needs to have $M$ net increase of 1).
Furthermore, we have $r=p_a+p_h r^2+p_0r$ since $\gap_n^{PLE}$ either increases straight away by one (with probability $p_a$),
or decreases by one (with probability $p_h$) in which case $\gap_n^{PLE}$ needs to increase by 2 to have a net increase of 1, or $\gap_n^{PLE}$ stays the same (with probability $p_0$) in which cases $\gap_n^{PLE}$ still needs to increase by 1.
% We have that $r=p+qr^2$~\cite{proba-stackexchange}.This gives us $r_M =\left( \frac{1-\sqrt{1-4pq}}{2q}\right)^M$.

We have  $r=p_a+p_hr^2+p_0r\Longleftrightarrow p_a+p_hr^2+(p_0-1)r=0$, with $p_a+p_h+p_0=1$. Hence, 
$r$ satisfies: $p_a+p_hr^2-(p_a+p_h)r=0$.
The solutions to this equation are $(1,p_a/p_h)$. Because $\alpha<1/2$, the random walk is transient with drift towards $-\infty$, hence $r<1$. This means that $r=p_a/p_h$
and $r_M =\left( \frac{p_a}{p_h}\right)^M=\left( \frac{e^{\alpha-1}-e^{-1}}{e^{-\alpha}-e^{-1}}\right)^M=\left( \frac{e^{\alpha}-1}{e^{1-\alpha}-1}\right)^M$.

The analysis in the SSLE case works the same but with $p_0 =0$, $p_a = \alpha$ and $p_h = 1-\alpha$. We then have
$r=\alpha+(1-\alpha) r^2$.
% %$r=p+q r^2$
% which gives
% $r=\frac{1+/-\sqrt{1+4\alpha(1-\alpha)}}{2(1-\alpha)}$
% with $\sqrt{1-4\alpha(1-\alpha)}=\sqrt{1-4\alpha+4\alpha^2}= |2\alpha-1|$.
The two solutions of this equation are 1 and $\frac{\alpha}{1-\alpha}$.
Since $r<1$, we have: $r_M =\left( \frac{\alpha}{1-\alpha}\right)^M$.
%$r=\frac{1-\sqrt{1-4pq}}{2p}$, with
% $\sqrt{1-4pq} = \sqrt{1-4\alpha(1-\alpha)}=\sqrt{1-4\alpha+4\alpha^2}= |2\alpha-1|= -\beta$.
% Hence $r_M =\left( \frac{1-(1-2\alpha)}{2(1-\alpha)}\right)^M=\left( \frac{\alpha}{1-\alpha}\right)^M$.

\end{proof}
Interestingly, 
since $x\mapsto (e^x -1)(1-x)$ is increasing on $[0,1/2]$,
we notice that for $\alpha<1/2$, $r_M^{PLE}<r_M^{SSLE}$. Starting from $-M$, the adversary is, hence, more likely to catch up the honest chain in the SSLE case than in the PLE case.
However, 
%since the gap is on expectation smaller in the SSLE case, 
this does not say anything about whether the adversary is more likely to win the private game of length $n$ as one process may be decreasing faster than the other.
We shall now compute the expected value of the gap in both cases to get a sense of their evolution.

Intuitively, since longest-chain protocols have been proven secure for an adversary that has less than half of the power~\cite{dembo2020everything}, the gap should decrease with time. The longer the chain is, the harder it is for an adversary to catch up with the honest chain and hence the bigger their disadvantage is, and hence their gap.
In the next lemma, we prove that the
expected gap is linear in $n$ and that the linear coefficient is bigger for the SSLE than PLE gap, and, therefore, the disadvantage of the adversary grows faster in the SSLE case, consistently with the intuition that SSLE is more secure.
Figure~\ref{fig:gap-coeff} show this coefficient for different values of $\alpha<1/2$.
\begin{lemma}
For every $n\in\mathbb{N}$ and $\alpha\in \; (0,1)$:
\begin{align*}
    \mathbb{E}[\gap_n^{SSLE}(\alpha)]&=(2\alpha-1)n \\
     \mathbb{E}[\gap_n^{PLE}(\alpha)]&=(e^{\alpha-1}-e^{-\alpha})n 
\end{align*}
\end{lemma}

\begin{proof}
Let $(R_n)_{n\in\mathbb{N}}$ denote a random walk that has a probability $p$ of going up by 1, a probability $q$ of going down by 1 and a probability $1-p-q$ of staying the same in each round.
Additionally we assume $R_0=0$.

It is trivial to verify that $X_n = R_n - (p-q)n$
is a martingale. 
Since martingales have constant expectations, we have $\mathbb{E}[X_n]=X_0 = 0$ and hence
$\mathbb{E}[R_n]=(p-q)n$.
We apply this result to $\gap_n^{SSLE}$ and $\gap_n^{PLE}$.
In the SSLE case, $p-q = \alpha - (1-\alpha)=2\alpha-1$ and in the
PLE case $p-q = p_a-p_h = e^{\alpha-1}-e^{-\alpha}$.
This proves the result. 
\end{proof}

\begin{figure}
     \centering
     \begin{subfigure}[b]{0.48\textwidth}
         \centering
    \includegraphics[width = \linewidth]{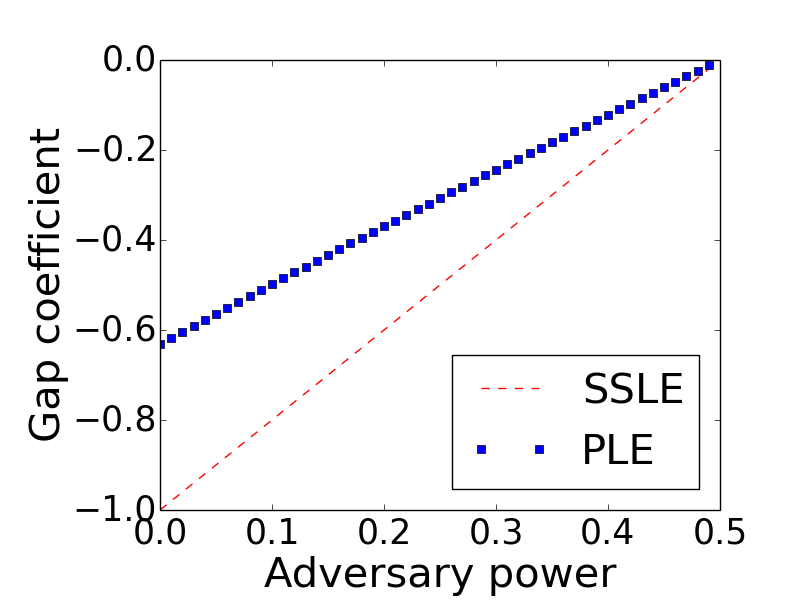}
    \caption{Linear  Coefficient for the Expected Gap}
    \label{fig:gap-coeff}
     \end{subfigure}
    %  \hfill
    %  \begin{subfigure}[b]{0.48\textwidth}
    %      \centering
    % \includegraphics[width = \linewidth]{figs/variance-coeff.png}
    % \caption{Variance Gap coefficient}
    % \label{fig:var-gap-coeff}
    %  \end{subfigure}
    %     \hfill
    %  \begin{subfigure}[b]{0.48\textwidth}
    %      \centering
    % \includegraphics[width = \linewidth]{figs/pred-int.png}
    % \caption{Maximum value of Gap with $k=7$}
    % \label{fig:pred-int}
    %  \end{subfigure}
\end{figure}

We now move on to prove the main result that gives the probability of success of the adversary in the SSLE game.

\begin{theorem}\label{thm:ssle}
The probability that the adversary wins the SSLE game for any
length greater or equal than $n$ is:
{\small
  \setlength{\abovedisplayskip}{6pt}
  \setlength{\belowdisplayskip}{\abovedisplayskip}
  \setlength{\abovedisplayshortskip}{0pt}
  \setlength{\belowdisplayshortskip}{3pt}
\begin{align*}
% \SSLE{\alpha}{n} &=\sum_{v=0,n+v\equiv 0 [2]}^n\binom{n}{\frac{1}{2}(n+v)}\alpha^{\frac{1}{2}(n+v)}(1-\alpha)^{\frac{1}{2}(n-v)}\\
%  & + \sum_{v=-n}^{-1} \binom{n}{\frac{1}{2}(n+v)}\alpha^{\frac{1}{2}(n-v)}(1-\alpha)^{\frac{1}{2}(n+v)} \\
\SSLE{\alpha}{n}=& \sum_{\substack{v=0,\\n+v\equiv 0 [2]}}^n \Bin(\alpha,n,\frac{1}{2}(n+v))
+\sum_{\substack{v=1,\\ n-v\equiv 0 [2]}}^{n} \Bin(1-\alpha,n,\frac{1}{2}(n-v))
\end{align*} }
\end{theorem}

\begin{proof}
For readability, in this proof we note $\gap_n$ instead of  $\gap_n^{SSLE}$. 
We start by noting that $-n\le\gap_n\le n$ for every $n\in\mathbb{N}$.
There are two scenarios where the adversary can have a private chain longer than the honest chain for a length of at least $n$: either the gap at $n$ is positive, or the gap is negative and the adversary catches up in the future (i.e., the gap reaches 0 in the future). This gives us the following:
\begin{align*}
    \SSLE{\alpha}{n} &= \Pr[\gap_n\ge 0] 
    +\Pr[\gap_n<0]\Pr[\gap_n \text{catches up}|\gap_n<0] \\
    &= \sum_{v=0}^{n}\Pr[\gap_n=v]
    +\sum_{v=1}^{n}\Pr[\gap_n=-v]r_{v}\\
    &=\sum_{\substack{v=0,\\n+v\equiv 0 [2]}}^n\binom{n}{\frac{1}{2}(n+v)}\alpha^{\frac{1}{2}(n+v)}(1-\alpha)^{\frac{1}{2}(n-v)} \\ & +
\sum_{\substack{v=1,\\ n-v\equiv 0 [2]}}^{n} \binom{n}{\frac{1}{2}(n-v)}\alpha^{\frac{1}{2}(n-v)}(1-\alpha)^{\frac{1}{2}(n+v)} \left(\frac{\alpha}{1-\alpha}\right)^{v}\\
&= \sum_{\substack{v=0,\\n+v\equiv 0 [2]}}^n\binom{n}{\frac{1}{2}(n+v)}\alpha^{\frac{1}{2}(n+v)}(1-\alpha)^{\frac{1}{2}(n-v)}\\ & +
\sum_{\substack{v=1,\\ n-v\equiv 0 [2]}}^{n} \binom{n}{\frac{1}{2}(n-v)}\alpha^{\frac{1}{2}(n+v)}(1-\alpha)^{\frac{1}{2}(n-v)} \\
&= \sum_{\substack{v=0,\\n+v\equiv 0 [2]}}^n \Bin(\alpha,n,\frac{1}{2}(n+v))
+\sum_{\substack{v=1,\\ n-v\equiv 0 [2]}}^{n} \Bin(1-\alpha,n,\frac{1}{2}(n-v))
\end{align*}
%http://cgm.cs.mcgill.ca/~breed/MATH671/lecture2corrected.pdf
\end{proof}

We now move on to prove the result in the PLE case. 
\begin{theorem}\label{thm:ple}
The probability that the adversary wins the PLE private game for any length greater or equal than $n$ is:
\begin{align*}& \PLE{\alpha}{n} = \sum_{v=0}^n \sum_{l=0}^{n-v} \Bin(p_0,n,l)\Bin(\frac{p_a}{1-p_0},n-l,\frac{1}{2}(n-l+v)) \\
   &+ \sum_{v = 1}^{n}\sum_{l=0}^{n-v} \Bin(p_0,n,l)\Bin(\frac{p_a}{1-p_0},n-l,\frac{1}{2}(n-l-v))\left(\frac{e^\alpha-1}{e^{1-\alpha}-1}\right)^{v}
\end{align*}
\end{theorem}

\begin{proof}

We use the same technique as in the previous theorem.

\begin{align*}
     &\PLE{\alpha}{n} = \Pr[\gap_n\ge 0] 
    +\Pr[\gap_n<0]\Pr[\gap_n \text{catches up}|\gap_n<0] \\
    &= \sum_{v=0}^{n}\Pr[\gap_n=v]
    +\sum_{v=1}^{n}\Pr[\gap_n=-v]r_{v}\\
   &= \sum_{v=0}^n \sum_{l=0}^{n-v} \Bin(p_0,n,l)\Pr[\gap^{SSLE}_{n-l}(\frac{p_a}{1-p_0})=v] \\
   &+ \sum_{v = 1}^{n}\sum_{l=0}^{n-v} \Bin(p_0,n,l)\Pr[\gap^{SSLE}_{n-l}(\frac{p_a}{1-p_0})=v]\left(\frac{e^\alpha-1}{e^{1-\alpha}-1}\right)^{v}\\
      &= \sum_{v=0}^n \sum_{l=0}^{n-v} \Bin(p_0,n,l)\Bin(\frac{p_a}{1-p_0},n-l,\frac{1}{2}(n-l+v)) \\
   &+ \sum_{v = 1}^{n}\sum_{l=0}^{n-v} \Bin(p_0,n,l)\Bin(\frac{p_a}{1-p_0},n-l,\frac{1}{2}(n-l-v))\left(\frac{e^\alpha-1}{e^{1-\alpha}-1}\right)^{v}\\
%   &=\sum_{v=0}^n \sum_{l=0}^v \Bin(p_0,n,l)\Bin(\frac{p_a}{1-p_0},n-l,\frac{1}{2}(n-l+v))\\ 
%   &+ \sum_{v = 1}^{n}\sum_{l=0}^v \Bin(p_0,n,l)\Bin(\frac{p_a}{1-p_0},n-l,\frac{1}{2}(n-l+v))(\frac{e^\alpha-1}{e^{1-\alpha}-1})^{-v}\\
%   &= \sum_{l=0}^{n-v} \Bin(p_0,n,l)[\Bin(\frac{p_a}{1-p_0},n-l,\frac{1}{2}(n-l))\\ 
%   &+\sum_{v=1}^n \Bin(\frac{p_a}{1-p_0},n-l,\frac{1}{2}(n-l+v))
%   \\ & +\Bin(\frac{p_a}{1-p_0},n-l,\frac{1}{2}(n-l-v))\left(\frac{e^\alpha-1}{e^{1-\alpha}-1}\right)^{v}]
\end{align*}
\end{proof}

% Based on the above theorem, we conclude that
% the probability of winning the SSLE game for any length greater or equal than $n$ is $\Omega((2n+1)\binom{n}{1/2n}\alpha^{n/2}(1-\alpha)^{n/2})$
\subsubsection*{Results interpretation}

In order to compare these two probabilities, we plot them for different values of $n$ and $\alpha$.
In Figure~\ref{fig:persistence}, we see that, as expected, SSLE performs much better than PLE: the adversary wins the PLE game with higher probability than the SSLE game.
To cite a few concrete examples, for $n=$ 300 and a 33\% adversary, the probability of success drops from $10^{-7}$ to $10^{-9}$ (Figure~\ref{fig:33-adv}).
For $\alpha = 0.33$ and $\epsilon = 10^{-12}$ the persistence parameter is $n_0=400$ in the SSLE case and 550 in the PLE case.
For a blockchain where one block is emitted every 30 seconds, breaking persistence with probability $10^{-12}$ would roughly occur once every million years.
% We also plot the results in the case of a 49\% adversary in Figure~\ref{fig:49-adv}. In this case, SSLE still performs notably better than PLE.

For all the values of $\alpha$ that we plotted, we see that, at least for $n$ big enough,
the probability of violating persistence is exponential (as the graph is linear in logarithmic scale), meaning that it is of the form $e^{-an}$ for $a\in\mathbb{R}^*_+$ in both the PLE and SSLE cases -- for different values of $a$ that we denote $a^{SSLE}$ and $a^{PLE}$.
We remark that this form is consistent with the bound found by Gazi et al.~\cite{gazi2020tight} and Li et al.~\cite{li2020close} in the context of Bitcoin
and is tighter than the one of $e^{-\Omega(\sqrt{n})}$ from Dembo et al.~\cite{dembo2020everything} which is based on a looser inequality (as they consider every attack possible).

For a fixed $\epsilon$ and large $n$, the persistence
parameter of the SSLE game, $n_0^{SSLE}$, can thus be expressed as follows:
$n_0^{SSLE}=\frac{a_{PLE}}{a_{SSLE}}n_0^{PLE}$, where $n_0^{PLE}$ is the $\epsilon$-persistence parameter of the PLE game.
Finding the value of $a_{PLE}$ and ${a_{SSLE}}$ is straightforward by computing a specific value of $\epsilon$ for some big enough $n$: $a= -\ln(\epsilon)/n$.
We find that the $\epsilon-$persistence parameter decreases by 
$17\%$ for a 49\% adversary, by roughly
$25\%$ for a 33\% or 25\% adversary, and by
$32\%$ for a 10\% adversary
in the SSLE case compared to the PLE case.
The improvement is substantial.

\begin{figure*}[h]
\centering
     \begin{subfigure}[b]{0.48\textwidth}
         \centering
  \includegraphics[width = \linewidth]{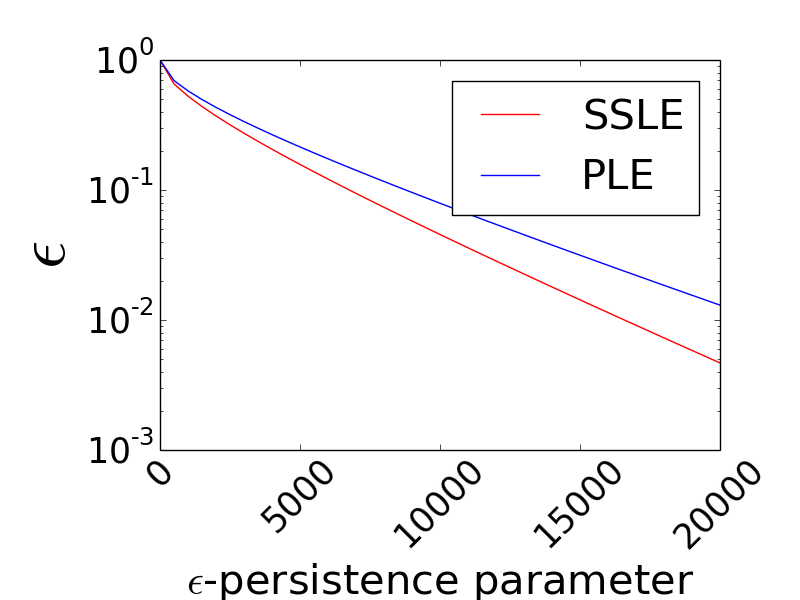}
\caption{$\alpha = 0.49$.}
\label{fig:49-adv}
     \end{subfigure}
     \hfill
     \centering
     \begin{subfigure}[b]{0.48\textwidth}
         \centering
 \includegraphics[width = \linewidth]{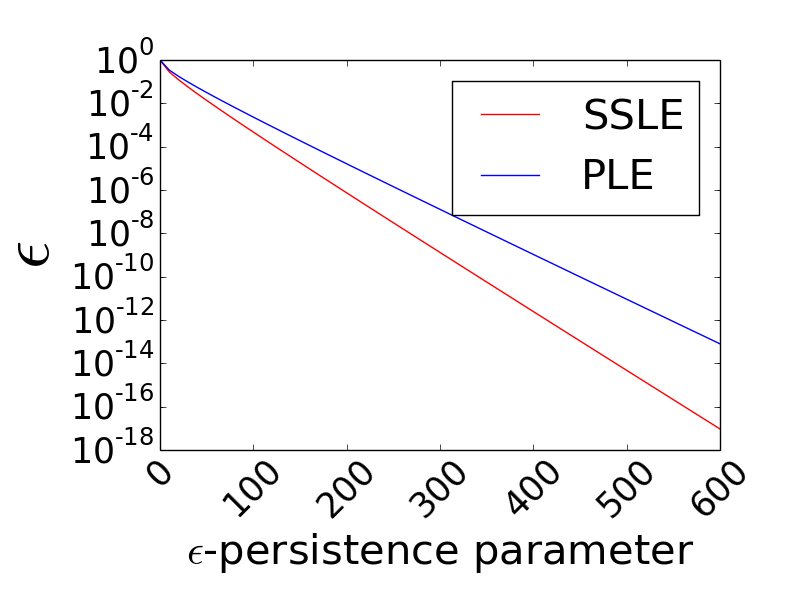}
\caption{$\alpha = 0.33$.}
\label{fig:33-adv}
     \end{subfigure}
     \hfill
     \begin{subfigure}[b]{0.48\textwidth}
         \centering
  \includegraphics[width = \linewidth]{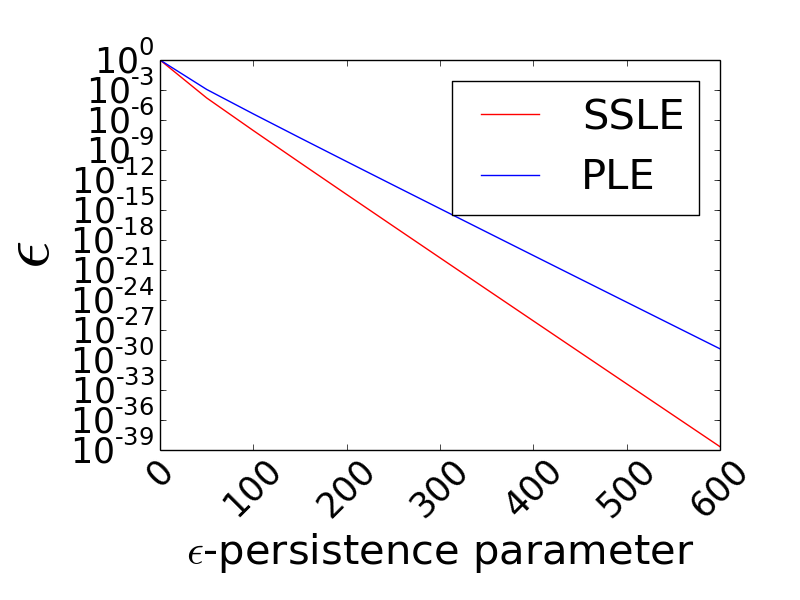}
\caption{$\alpha = 0.25$.}
\label{fig:25-adv}
     \end{subfigure}
     \hfill
     \centering
     \begin{subfigure}[b]{0.48\textwidth}
         \centering
 \includegraphics[width = \linewidth]{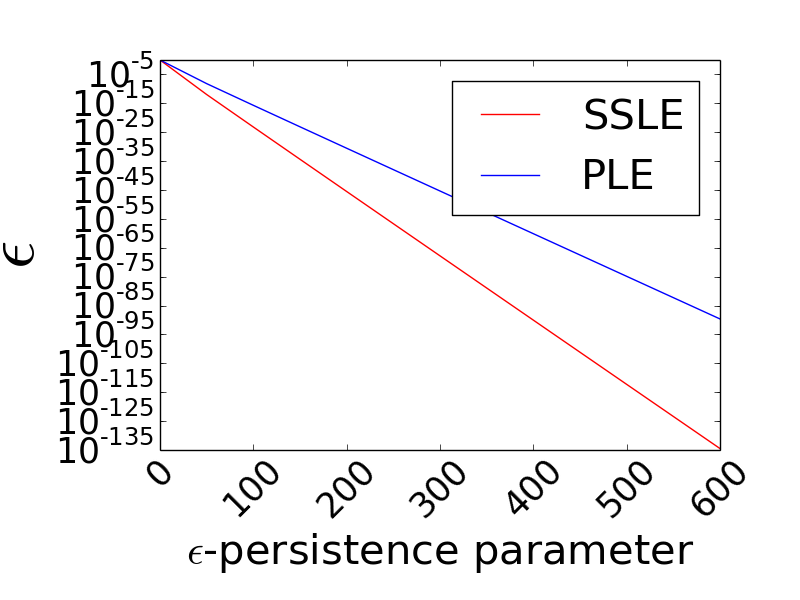}
\caption{$\alpha = 0.1$.}
\label{fig:33-adv}
     \end{subfigure}
\caption{$\epsilon-$persistence parameter for an adversary with power $\alpha$ (logarithmic scale).}
\label{fig:persistence}
\end{figure*}

% One important limitation of our previous analysis, is that it assumes access to a perfect source of randomness. We relax this assumption
% in the next section.

\section{Grinding analysis}\label{sec:grinding}

In the previous section, we 
have assumed that the random beacon given as input to the leader election at each round came from a perfect source of randomness. Such an assumption can be achieved with a decentralized random beacon~\cite{syta2017scalable}, for example. 
However most PoS protocols rely on an internal random beacon, which is
based on the state of the chain. 
Such random beacons are vulnerable, to some extent, to grinding attacks where an adversary grinds through the block space in order to bias the randomness and find a value for which an unfair advantage can be extracted.
Bagaria et al.~\cite{nakamoto-pos} studied
the security of PoS protocols where the randomness is updated every $c$ blocks and the trade-offs between updating the randomness more frequently (and being more vulnerable to grinding) or less frequently (and being more predictable).
In this section, we study this specific case and 
compare the security guarantees of PoS blockchains against private attacks that use grinding in the PLE and SSLE case.
We start by presenting this new model in Section~\ref{sec:grinding-model} before moving on to the analysis in Section~\ref{sec:grinding-analysis}.

% We consider a model where the randomness at round $i$, $r_i$,
% is a function of the state of the chain before round $i$ (e.g., the randomness in the previous block at round $i-1$), the elected leader signature and the current round~\cite{praos}.
% This randomness is given as an input to the leader election function that will, pseudo-randomly elect some leaders.

\subsection{Model}\label{sec:grinding-model}

\subsubsection{Random Beacon and Grinding}
%When grinding the adversary can create a number of alternative chains at the same time.
%We now study what is the gain in security, in the particular case of grinding attacks, when using SSLE compared to PLE.
We assume that the random beacon given as input to the leader election in each round is based on the data 
in the block created in the previous round.
%is updated every block in a limited grindable way.
A block created by different miners or a round skipped will all produce different beacons and thus election results in the next round but two blocks created by the same leader at the same round with different content (e.g., different transaction sets) will produce the same random beacon.

As a concrete, simplified, example, we assume that a random beacon $r_0$ is initialized using a multiparty computation protocol~\cite{beimel2010protocols,ben1985collective,cascudo2017scrape} at round 0. Randomness at round $i$ is then defined as 
$r_i = \sigma_{sk}(r_{i-1}\oplus i)$,
where $\sigma$ is the deterministic signature of
the player creating the block in round $i$. In practice a Verifiable Random Function~\cite{micali1999verifiable} will be used instead of a signature scheme but this does not matter for our analysis.
If no block is created at round $i-1$ because the elected leader was offline or no leader was elected (in the PLE case),
then the previous randomness is used, i.e.,
$r_i =\sigma_{sk}(r_{i-2}\oplus i)$.

This scenario corresponds to $c=1$ in ~\cite{nakamoto-pos} and it could be extended to the more general case where the randomness is updated every $c$ blocks instead (i.e.,
$r_i = \sigma_{sk}(r_{i-(i\%c)}\oplus i)$)
as Bagaria et al.~\cite{nakamoto-pos} showed.
In the case where $c>1$ 
the grinding is more limited but the protocol is more predictable which is an undesirable property (that we ignore in this work). As far as grinding is concerned $c=1$ corresponds to the worst-case scenario.

The grinding works as follows: whenever the adversary is elected leader, it can decide to publish its block or  skip the round, thus biasing the randomness.
Furthermore, in the PLE case (where the adversary can potentially produce more than one block per round), the adversary could decide which of its blocks to use, if it was elected more than once, or simply skip this round.
By trying out different combinations of blocks or skipping rounds, the adversary may find a particular chain where it is luckier and create a longer chain.

\subsubsection{Branching Random Walks}
Similarly to Bagaria et al.~\cite{nakamoto-pos}, we use the theory of branching random walks to study the problem of grinding. The rest of the assumptions (e.g., about synchronicity) are as
presented in Section~\ref{sec:model}.

We consider the following branching random walk that we note BRW. We first define it using standard branching process vocabulary, before explaining how it relates to grinding and blockchains.
As before, time is divided into discrete time steps.
An initial ancestor is located at the origin. 
At each time step a particle gives birth to a random number of children before dying.
Each child is randomly scattered through $\mathbb{N}$. In the next time-step, each child will give birth to their own children before dying.
The process repeats indefinitely. Each particle is independent of the others and verifies the following properties:

\begin{enumerate}
    \item Each particle has at least one child.
    The first child is located at the same position as its parent.
    \item $v$ denotes the position of a particle. The other children of the particle, if they exist, are located at position $v+1$.
    \item The number of children located at position $v+1$ follows a distribution noted Z.
\end{enumerate}

% \begin{itemize}
%     \item its number of children is 1 with probability $1-\alpha$ and 2 with probability $\alpha$.
%     \item the first child of every particle stays at the same position of its parent.
%     \item let's note $v$ the position of a particle. If the particle has a second child, then this child has a position $v+1$.
% \end{itemize}
% The process repeats indefinitely. Each particle is independent of the others.

% In the second setting, BRW2, for every particle we have the following:
% \begin{itemize}
%     \item its number of children is 1 with probability $e^{-\alpha}$ and $i> 1$ with probability $\frac{e^{-\alpha} \alpha^{i-1}}{(i-1)!}$.
%     \item the first child of every particle stays at the same position of its parent.
%     \item let's note $v$ the position of a particle. If the particle has more than one child, then each child has a position $v+1$, except the first child that has a position $v$.
% \end{itemize}

Such a process is illustrated in Figure~\ref{fig:brw}.
We will consider two different distributions for $Z$.
In the first case, $Z$ will follow a Bernoulli distribution of parameter $\alpha$. There is exactly one particle at height $v+1$ with probability $\alpha$ (and 0 otherwise). 
In this case, each particle will have one child with probability $1-\alpha$ and two children (one at position $v$ and one at position $v+1$) with probability $\alpha$. We denote by BRW1 the associated branching random walk.

In the second case, $Z$ follows a Poisson distribution of parameter $\alpha$. There will be $i$ particle in position $v+1$ with probability $\frac{e^{-\alpha} \alpha^{i}}{i!}$.
In this case, each particle has one child with probability $e^{-\alpha}$ and a total of $i>1$ children with probability $\frac{e^{-\alpha} \alpha^{i-1}}{(i-1)!}$.
We denote by BRW2 this branching random  walk.

Intuitively, a particle having exactly one child corresponds to the case where the adversary is not elected leader hence its chain does not increase by one and the child particle stays at the same position as its parent.
This happens with probability $1-\alpha$ in the SSLE case and $e^{-\alpha}$ in the PLE case.
A particle having more than one child corresponds to the case where the adversary was elected leader in that round, in which case its private chain increases by one and, analogously, the position of the other children increases by one.
% The number of children hence follows either a Bernoulli distribution, in the SSLE case, or a Poisson distribution, in the PLE case.

In the case where the adversary has $m>1$ of its miners elected leader, then all of them will create a new block with a new random value and thus a new potential chain. Each particle therefore corresponds to a new chain that will grow independently of the other from then on.
The maximum position for BRW at time $n$ corresponds to the longest chain that the adversary has been able to create by grinding. 
We are interested in comparing this value with the
length of the honest chain. 

Honest players do not form a coalition and are not grinding, hence their chain evolves as a random walk that increases by one with probability $\delta_+$ and stays unchanged with probability $\delta_0 = 1- \delta_+$.
In the SSLE case, we have $\delta_+^{SSLE}=1-\alpha$
and in the PLE $\delta_+^{PLE}=1-e^{\alpha-1}$.
\begin{figure}[h]
\includegraphics[width = \linewidth]{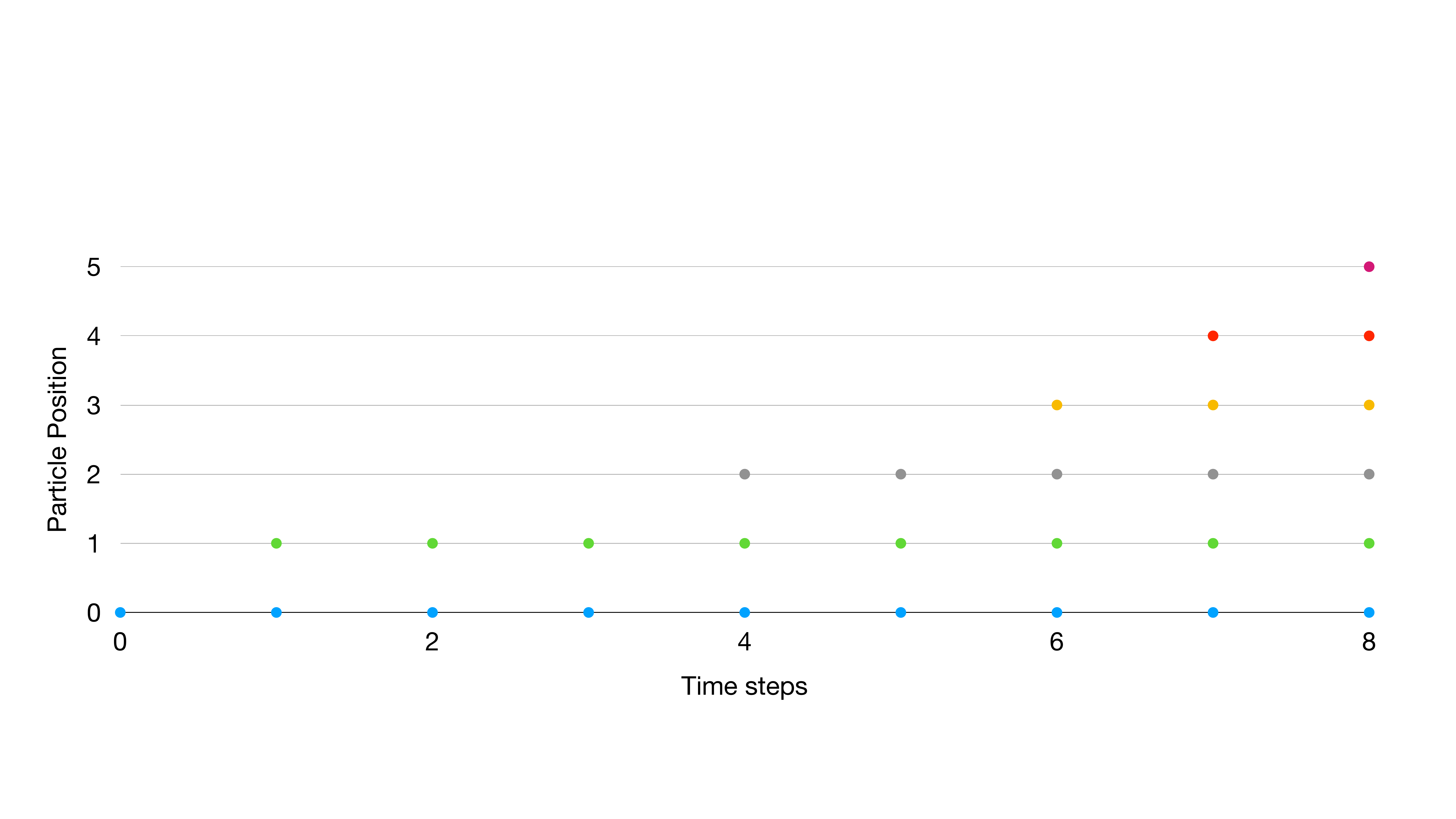}
\caption{Example of a branching random walk. There may exist more than one particle at each position.}
\label{fig:brw}
\end{figure}
Let $M_i$ denote the maximum position of the generic branching random walk BRW described above at time $i\in\mathbb{N}$.
When $Z$ is specified to be a Bernoulli distribution of parameter $\alpha$ we will write this process $M^{SSLE}(\alpha)$ and when $Z$ is a Poisson distribution of parameter $\alpha$, we will write the corresponding stochastic process $M^{PLE}(\alpha)$.
Similarly, we denote by $S_i$ the position of the generic honest random walk at time $i$ and will specify $S_i^{SSLE}(\alpha)$ or $S_i^{PLE}(\alpha)$ when needed.

The SSLE and PLE grinding games are then defined as follows. 
\begin{definition}[$(\alpha,L)$-SSLE Grinding  Game]
For $M^{SSLE}(\alpha)$ and $S^{SSLE}(\alpha)$
as defined above,
we say that the adversary wins the SSLE grinding game of length $L$ and power $\alpha$ if at timestep $L$, $M_L^{SSLE}(\alpha)\ge S_L^{SSLE}(\alpha)$.
\end{definition}

\begin{definition}[$(\alpha,L)$-PLE Grinding  Game]
For $M^{PLE}(\alpha)$ and $S^{PLE}(\alpha)$
as defined above,
we
say that the adversary wins the PLE grinding game of length $L$ and power $\alpha$ if at timestep $L$,
$M_L^{PLE}(\alpha)\ge S_L^{PLE}(\alpha)$.
\end{definition}

As before, we will sometimes refer to simply \emph{the grinding game} when talking about the generic processes $M$ and $S$.
% \begin{theorem}
% \[\mathbb{E}[Z_n^{SSLE}]= \mathbb{E}[Z_n^{PLE}]= (1+\alpha)^n\]
% \end{theorem}

% \begin{proof}
% \url{https://www.esaim-proc.org/articles/proc/pdf/2011/01/proc113101.pdf}
% \end{proof}
% As before we see that the adversary wins the SSLE, resp. PLE, grinding game if and only if the provate attack on length $L$ is successful.

We are interested
in comparing the persistence parameter for SSLE compared to PLE and, hence, comparing the probabilities of winning the SSLE grinding game and the PLE grinding game.

Additionally, we are interested in the security threshold of the two grinding games. In the non-grinding case, we already know that the protocol is secure against private attacks if and only if $\alpha<0.5$ for obvious reasons (i.e., the 
$\epsilon-$persistence parameter exists and is finite for every $\epsilon\in (0,1)$ and $\alpha<0.5$).
With the grinding attack, an adversary may win an unfair advantage in the private attack and take over the honest chain of any length even with less than half of the stake.
For the PLE case, Bagaria et al.~\cite{nakamoto-pos}
proved that the protocol is secure in the case of grinding if and only if $\alpha<1/(1+e)$ for a network delay $\Delta=0$.
In the next section, we first derive the
security threshold, i.e., the
biggest value $\alpha_0$ such that there exists a finite $\epsilon-$persistence parameter for every $\epsilon\in(0,1)$.
%an adversary will win the SSLE grinding game of parameter $\alpha_0$ and $L$ almost surely for $L$ big enough and compare it to the PLE case.
We then look at the probabilities of winning the game of length $n$ for an adversary with power less than $\alpha_0$ and again, compare the two cases.

Before moving on to the analysis, we make one important remark in the SSLE case with grinding.
Since the randomness on the honest and adversarial chains are now different, the leader elections on each of these chains become independent. This is unlike the non-grinding case, where we had $h_n = 1- a_n$.
The SSLE thus does not act like an SSLE anymore. There could be more than one winner per round, each on a different chain, or no block created at all in some rounds (e.g., if in the adversarial chain the leader elected is honest and in the honest chain the leader elected is adversarial).
The SSLE thus acts more like a probabilistic leader election where one leader is elected with probability $\alpha$ at each round on each adversarial chain, and one leader is elected with probability $1-\alpha$ on the honest chain. Unlike with the PLE private game, however, there can be at most one leader elected per round on each chain.

\subsection{Analysis}\label{sec:grinding-analysis}
\subsubsection{Security Threshold}

We start by defining the security threshold of the grinding game. Intuitively, the security threshold captures the threshold for which the protocol is secure after a certain length. In other words, the 
probability of winning the game
can be made as small as desired by choosing a long enough length.

\begin{definition}[Security threshold]
If there exists $\alpha_0\in\; (0,1)$ such that for every $\alpha>\alpha_0$, 
$M_i(\alpha)> S_i(\alpha)$ asymptotically a.s. and for $\alpha<\alpha_0$,
$S_i(\alpha)> M_i(\alpha)$ asymptotically a.s., we call such $\alpha_0$ the \emph{security threshold} of the grinding game.
\end{definition}

We note that the above definition does not say anything about the behaviour of the game with power exactly $\alpha_0$.
In order to compute the security threshold in both cases, we use results from Biggins~\cite[Section 6]{biggins1976first} that we adapt to our discrete time model.
% We note $Z_i$ the random variable that describes the number of children at position 
Specifically, Biggins considers a continuous model where particles are scattered through $\mathbb{R}$ and particles 
are of different types, whereas we only have one type of particle that is scattered through $\mathbb{N}$.
Additionally, he considers the minimum position of the branching random walk rather than the maximum. With this mind, we can directly derive the following result:
\begin{align*}
    &\lim_{n\to\infty} \frac{M_n}{n} \rightarrow -\gamma \text{ a.s. where $\gamma$ is defined as below:}\\
    % &\phi(\theta) = \mathbb{E}[e^{-\theta \cdot 0}(-Z_0)+e^{-\theta \cdot 1}(-Z_1)] \\
&\phi(\theta) := \mathbb{E}[1+e^{\theta  }Z] \\
   &   \mu(a) := \inf \{ e^{\theta a} \phi(\theta): \theta \ge 0 \} \\
   & \gamma := \inf \{a: \mu(a) \ge 1\}
\end{align*}

Similarly following the law of large numbers (applied to the Binomial distribution), we have $\lim_{n\to\infty}~\frac{S_n}{n}\rightarrow \delta_{+}$ a.s.
Hence, we conclude that $M_n > S_n$ a.s. asymptotically if $-\gamma > \delta_+$
and $M_n < S_n$ a.s. asymptotically if $-\gamma < \delta_+$.
As a consequence, the security threshold corresponds to the case $-\gamma = \delta_+$.
Based on this observation, we prove the following theorem.

\begin{theorem}
The security threshold of the SSLE, resp. PLE, grinding games are: $\alpha^{SSLE}\simeq 0.36$ and
$\alpha^{PLE}\simeq 0.265$.
\end{theorem}
\begin{proof}
In order to prove the theorem, we compute the value of $\gamma$. We start by computing $\phi$.
\begin{align*}
    % \phi &= \mathbb{E}[e^{-\theta t}(-Z)]\\
    % \phi &= \mathbb{E}[1+e^{-\theta \cdot 1}(-Z)]\\
     \phi &= \mathbb{E}[1+e^{\theta} Z]\\
      &= 1+e^{\theta} \alpha
\end{align*}

We first note that this expression is independent of whether the number of blocks (or children at position 1) follows a Bernoulli or Poisson distribution since they both have the same expectation.
This is in itself an interesting observation since it means that, asymptotically, both processes $M_i^{SSLE}$
and $M_i^{PLE}$ behave similarly (although this is not the case for $S_i^{SSLE}$ and $S_i^{PLE}$).

We now move on to compute $\mu$ and $\gamma$, which will be equal for both branching random walks since they only depend on $\phi$.

First, we compute $\mu(a)= \inf \{e^{\theta a}\phi(\theta):\theta\ge 0\}$.
We have $e^{\theta a}\phi(\theta)    = e^{\theta a}(1+e^{\theta} \alpha)$, hence, $\mu(a) = 1+\alpha$ for $a\ge 0$, $\mu(a) = 0$ for $a\le-1$, so it remains to compute $\mu(a)$ for $-1<a<0$.
We have $\frac{\partial}{\partial \theta}(e^{\theta a}\phi(\theta))= e^{\theta a}(a+\alpha(a+1)e^\theta)\ge 0 \Leftrightarrow \theta \ge \ln(-\frac{a}{\alpha(a+1)})$.
 Hence $\theta\mapsto e^{\theta a}\phi(\theta)$ is decreasing up until $\theta_0 = \ln(-\frac{a}{\alpha(a+1)})$ and increasing after this.
We therefore have that for $-1<a<0$, $\mu(a)= e^{\theta_0 a}\phi(\theta_0)= (\frac{-a}{\alpha (a+1)})^a\frac{1}{a+1}$.
  \[
    \mu(a) = \left\{\begin{array}{lr}
        1+\alpha, & \text{for } a\ge 0\\
        \left(\frac{-a}{\alpha (a+1)}\right)^a\frac{1}{a+1}, & \text{for } -1<a<0\\
        0, & \text{for } a<-1
        \end{array}
        \right.
  \]

 Next we compute $\gamma = \inf \{ a: \mu(a)\ge 1 \}$.
 We denote $\zeta(a) = (\frac{-a}{\alpha (a+1)})^a\frac{1}{a+1}$ for $-1<a<0$.
 We have $\frac{\partial}{\partial a}\zeta = \frac{1}{a+1}(\frac{-a}{\alpha(a+1)})^a\log(\frac{-a}{\alpha(a+1)})$ and $\frac{\partial}{\partial a}\zeta \ge 0 \Leftrightarrow a \le -\frac{\alpha}{\alpha +1}$. Hence $\zeta$ is first increasing then decreasing.
 Furthermore, we have that $\lim_{a\to 0}\zeta(a)= 1$ and $\lim_{a\to -1}\zeta(a)= \alpha$ hence 
 $\zeta$ reaches one exactly once and $\gamma$ is the
 unique solution to $\zeta(a)=1$ for $-1<a<0$.
We thus know that gamma solves the following equation:

\[\left(\frac{-\gamma}{\alpha (\gamma+1)}\right)^\gamma\frac{1}{\gamma+1}=1\]

In the SSLE case, the threshold corresponds to the case where
$-\gamma = 1-\alpha^{SSLE} $ and, in the PLE case, the threshold corresponds to the case $-\gamma = 1- e^{\alpha^{PLE} -1}$.
Hence $\alpha^{SSLE}$ satisfies the following equation:
\[
\left(\frac{1-\alpha^{SSLE}}{(\alpha^{SSLE})^2}\right)^{\alpha^{SSLE}-1}=\alpha^{SSLE}
\]

The solution can be computed numerically giving $\alpha^{SSLE} \simeq 0.360$.

$\alpha^{PLE}$ on the other hand, satisfies the following equation:

% \[
% (e^{\alpha^{PLE}-1}-1)\ln\left(\frac{1-e^{\alpha^{PLE}-1}}{\alpha^{PLE}e^{\alpha^{PLE}-1}}\right)=\alpha^{PLE}-1\]
\[
\left(\frac{1-e^{\alpha^{PLE}-1}}{\alpha e^{\alpha^{PLE}-1}}\right)^{e^{\alpha^{PLE}-1}-1}=e^{\alpha^{PLE}-1}\]

which can be solved numerically giving $\alpha^{PLE} \simeq 0.265$.
\end{proof}

The above theorem shows that SSLE significantly increases the security threshold in the grinding private attack.
We also remark that the threshold found in the PLE case is similar to the one found by Bagaria et al.~\cite{nakamoto-pos} ($1/(1+e)$), although there is a slight difference between the models. 
The reference considers a continuous-time model (i.e., the slot duration $\delta$ is very small) and a delay of propagation of zero. This means that, for example, if two honest leaders were to find a block at $\delta$ milliseconds of interval, where $\delta$ is very small, then these two blocks will be added to the honest chain even though in practice they were found at the same time.
In our model, we consider a discrete time model and a synchronous network with a strictly positive delay, and so two blocks found in the same round cannot be added to the same chain.

\begin{figure*}
     \centering
     \begin{subfigure}[b]{0.48\linewidth}
         \centering
    \includegraphics[width = \linewidth]{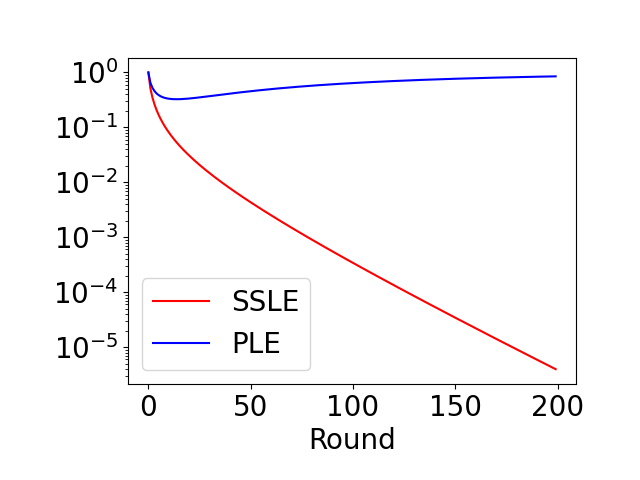}
    \caption{$\alpha = 0.3$}
    \label{fig:grind49}
     \end{subfigure}
     \hfill
     \begin{subfigure}[b]{0.48\linewidth}
         \centering
    \includegraphics[width = \linewidth]{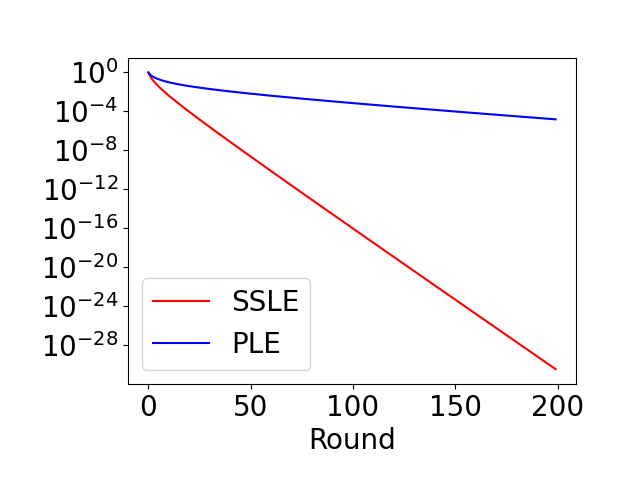}
    \caption{$\alpha = 0.2$}
    \label{fig:grind33}
     \end{subfigure}
        \hfill
             \begin{subfigure}[b]{0.48\linewidth}
         \centering
    \includegraphics[width = \linewidth]{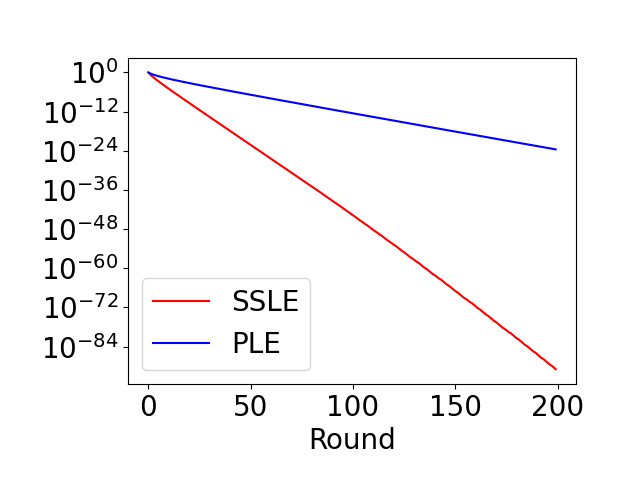}
    \caption{$\alpha = 0.1$}
    \label{fig:grind33}
     \end{subfigure}
     \hfill
     \caption{Probability that the adversarial chain is greater or equal than the honest chain when grinding for a chain of length $n$}
     \label{fig:grinding}
\end{figure*}

\subsubsection{Persistence parameter}

\begin{figure*}
\centering
     \begin{subfigure}[b]{0.48\textwidth}
         \centering
  \includegraphics[width = \linewidth]{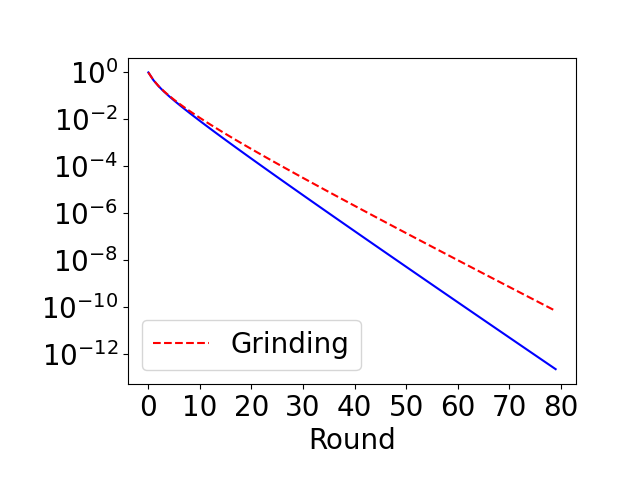}
\caption{PLE with and without grinding: $\alpha = 0.1$.}
\label{fig:grinding-ple-10}
     \end{subfigure}
     \hfill
     \centering
     \begin{subfigure}[b]{0.48\textwidth}
         \centering
 \includegraphics[width = \linewidth]{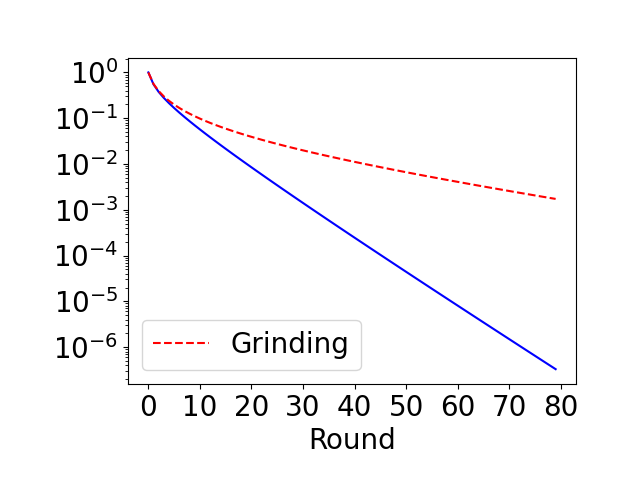}
\caption{PLE with and without grinding: $\alpha = 0.2$.}
\label{fig:grinding-ple-20}
     \end{subfigure}
     \hfill
     \begin{subfigure}[b]{0.48\textwidth}
         \centering
  \includegraphics[width = \linewidth]{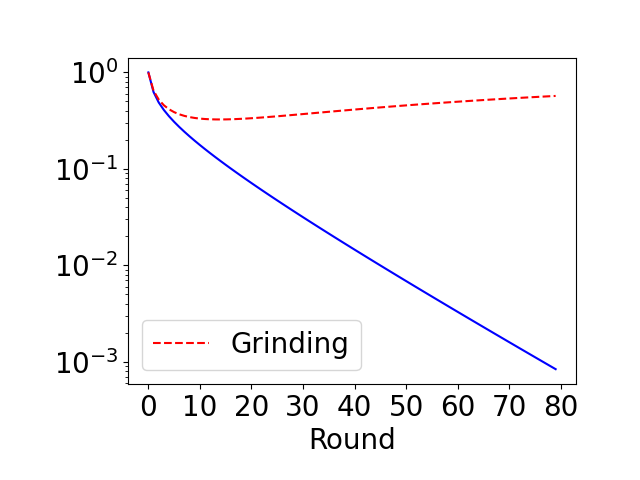}
\caption{PLE with and without grinding: $\alpha = 0.3$.}
\label{fig:grinding-ple-30}
     \end{subfigure}
     \hfill
     \centering
     \begin{subfigure}[b]{0.48\textwidth}
         \centering
 \includegraphics[width = \linewidth]{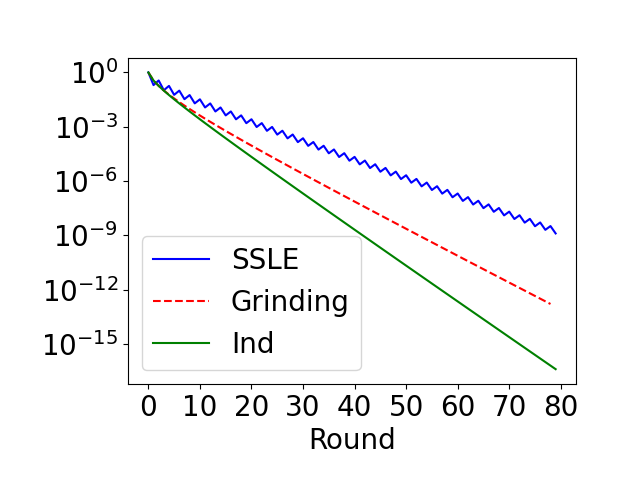}
\caption{SSLE: $\alpha = 0.1$.}
\label{fig:grinding-ssle-10}
     \end{subfigure}
          \hfill
     \centering
     \begin{subfigure}[b]{0.48\textwidth}
         \centering
 \includegraphics[width = \linewidth]{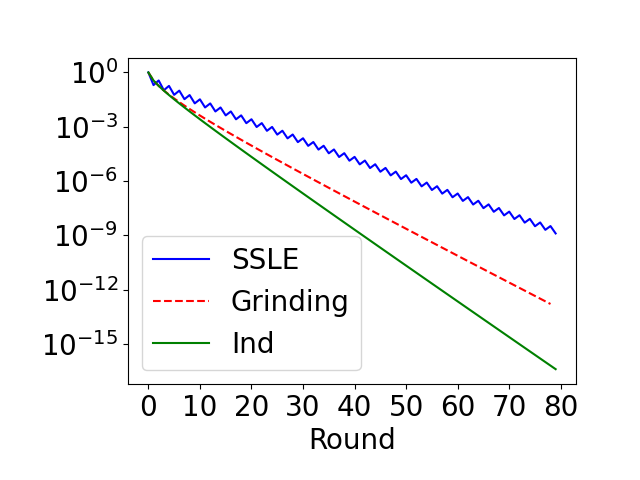}
\caption{SSLE:$\alpha = 0.2$.}
\label{fig:grinding-ssle-20}
     \end{subfigure}     \hfill
     \centering
     \begin{subfigure}[b]{0.48\textwidth}
         \centering
 \includegraphics[width = \linewidth]{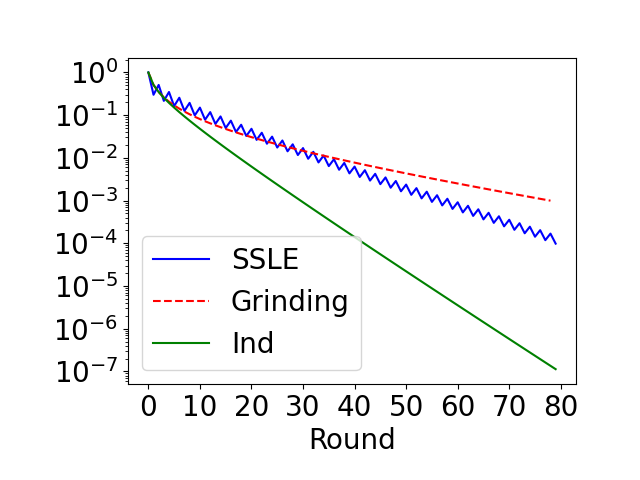}
\caption{SSLE: $\alpha = 0.3$.}
\label{fig:grinding-ssle-30}
     \end{subfigure}
\caption{Probability of winning the simple or grinding private games of length exactly $n$ for different values of $\alpha$ (logarithmic scale).
``Ind'' corresponds to the Independent SSLE private game.}
\label{fig:grinding-vs-non-grinding}
\end{figure*}
In the previous section, we computed the security threshold of the SSLE and PLE grinding games. We therefore know that, for an adversary below that power, we can find a length $L$ such that the probability of winning the grinding game of parameter $(\alpha,L)$ is as small as desired.
However, our analysis was only asymptotic and did not give any information about the persistence parameter or the behaviour of the game for a shorter time period.
In this section, we study the probability of winning the SSLE and PLE grinding games of length $n$ in order to give an estimate of the $\epsilon-$persistence parameter.

We are interested in the variable $(D_i=M_i-S_i)_{i\in\mathbb{N}}$ and especially the event
$M_i-S_i\ge 0$ that corresponds to the adversarial chain being longer or equal than the honest chain and hence the adversary winning the grinding game.

We denote $a_{i,j}=P(M_i< j)$.
Since, by definition, $M_0=0$, we have $a_{0,j}=1$ for $j> 0$.
We now define the following recursive formula for $a_{i,j}$:
{\small 
  \setlength{\abovedisplayskip}{6pt}
  \setlength{\belowdisplayskip}{\abovedisplayskip}
  \setlength{\abovedisplayshortskip}{0pt}
  \setlength{\belowdisplayshortskip}{3pt}
  \begin{align*}
          a_{i,j} &=  \sum_{m=0}^{\infty}P(Z=m)\Pr[M_i< j |  Z=m]\\
    &=  a_{i-1,j}\sum_{m=0}^{\infty}P(Z=m)(a_{i-1,j-1})^m
  \end{align*}
}%
This equality is explained as follows.
We start at time $0$.
We denote $m$ the number of children at position 1 of the initial particle (i.e., the total number of children is $m+1$).
The first child stays at position 0, whereas the other $m$ children will increase position by one.
All of the children generate independent processes similar to their ancestor, except starting at $(i+1,0)$ for the first child, and $(i+1,1)$ for the other $m$ children.
The process $M$ will not reach $j$ at step $i$ if and only if none of the processes engendered 
by the children of the original particle reach $j$. Since all the processes are independent, the probability that the process engendered by the first child never reaches $j$ is $a_{i-1,j}$. For the rest of the $m$ children,
this probability is $a_{i-1,j-1}$.
Conditional on the particle having $m+1$ children, the probability that $M$ does not reaches $j$ by time $i$ is equal to $a_{i-1,j}(a_{i-1,j-1})^{m}$.

Adapting the above probabilities to the SSLE grinding game yields the following:
\begin{align*}
    a_{i,j}^{SSLE}= a_{i-1,j}^{SSLE}(1-\alpha + \alpha a_{i-1,j-1}^{SSLE})
\end{align*}

Whereas in the PLE game, it becomes:
\begin{align*}
    a_{i,j}^{PLE}=  a_{i-1,j}^{PLE}e^{\alpha(a_{i-1,j-1}^{PLE}-1)}
\end{align*}

We note we have found a recursive formula for $a_{i,j}$ that depends on $(a_{i-1,j-1},a_{i-1,j})$.

Next, we are interested in
$b_{i}=\Pr[Mi-S_i\ge0]$.
We have:
\begin{align*}
   b_{i}&= \Pr[Mi-S_i\ge0]\\
   &= \sum_{s=0}^{i}\Pr[S_i=s]\times\Pr[M_i\ge s]\\
   &= \sum_{s=0}^{i}\Pr[S_i=s]\times(1-a_{i,s})\\
   &= \sum_{s=0}^{i}\Bin(\delta_+,i,s)\times(1-a_{i,s})
\end{align*}

 We plot this probability for different values of $\alpha$ and $n$ in Figure~\ref{fig:grinding}.
Here, we quickly remark that this probability does not allow us to find
the exact $\epsilon-$persistence parameter
as it is not the probability
that the adversary violates the persistence of the blockchain for any length greater than $n$ but the probability that the adversary wins the game of length exactly $n$. 
%In order to compute this exact probability, we would need to compute the probability that the adversary wins the SSLE, resp. PLE, game for any $n\ge n_0$.
In theory, the probability that the adversary wins the game for any $n\ge n_0$ is slightly bigger than the probability we computed as there is always a small chance that an adversary that did not win at length $n$ could catch up in the future.
In the grinding case, this probability is much more complex to compute than in the previous section. However, the probability that the adversary wins the grinding game of length exactly $n_0$ still provides an interesting proxy measure for the security of the underlying protocol.

% We see that in all cases the adversary performs better in the PLE game than in the SSLE case, confirming that SSLE is more secure.
% More specifically, we see that the attack will always be successful, even for a 25\% adversary whereas in the PLE case, the attack quickly becomes unsuccessful.
% Hence SSLE lowers the security threshold of PoS longest chains against private grinding attacks.

\paragraph{Results Interpretation}
In Figure~\ref{fig:grinding} we see that SSLE performs consistently better than PLE in the sense that the probability of winning the grinding game is smaller for SSLE than for PLE.
 We also notice, as before, that for $n$ big enough,
 this probability can be approximated as $e^{-an}$.
 Using the same method as before, we can find that SSLE 
 reduces the persistence parameter by roughly 70\% in the case of a 10\% adversary and 80\% for a 20\% adversary.
 In this case, the improvement is even more drastic than in the private game. Unlike the private game, however, the reduction is more noticeable for a 20\% than for a 10\% adversary.
 
It is also interesting to compare the probabilities of winning the private game vs winning the grinding game. We plot these probabilities in Figure~\ref{fig:grinding-vs-non-grinding}. 
In the simple (i.e., non-grinding) case, we compute the probability that the gap is positive instead of using the probability in Section~\ref{sec:analysis}, as this matches the probability we have computed in the grinding case.
Intuitively, grinding should increase the probability of winning the game of length $n$, which is what we observe for PLE in Figure~\ref{fig:grinding-ple-10}, \ref{fig:grinding-ple-20} and~\ref{fig:grinding-ple-30}.
However, in the SSLE case, we observe the opposite for a 10 and 20\% adversary.
As we have discussed before, in the case of grinding, the SSLE game does not act anymore as a single secret leader election since the adversarial and honest chains are now independent as they operate on different random beacons.
The SSLE grinding game is thus more similar to the PLE private game than to the SSLE private game, except that the probabilities $p_a$ and $p_h$ should be adapted accordingly. 
We now define the following game:
\begin{definition}[$(L,\alpha)$-idependent SSLE Private Game]
The independent SSLE private game with parameters $(L,\alpha)$ is defined as follows:
at each round $n\in [1,\dots,L ]$
a number $a_n$ of adversarial leaders is selected from a  Bernoulli distribution of parameter $\alpha$
and a number $h_n$  of honest leaders is selected from an independent  Bernoulli distribution of parameter $1-\alpha$.
We say that the adversary wins the PLE private game of length $L$ and power $\alpha$ if the number of rounds with non-zero adversarial leaders is greater or equal than the number of rounds with non-zero honest leaders, i.e.:
$$|\{n \in [1,\dots,L ]: a_n=1\}|\ge |\{n \in [1,\dots,L ]: h_n=1\}|.$$
% and we note $\PLE{L}{\alpha}$ the associated probability.

\end{definition}
% It thus make more sense to compare the SSLE grinding game to a private game
% Hence in order to compare the grinding to non-grinding case, it makes more sense to compare it to a case where the probabilities of adversarial and honest leaders are independent and each following a Bernoulli distribution of parameters $\alpha$ and $1-\alpha$. We call this game the \emph{independent SSLE private game}.
This game is equivalent to a PLE private game, except that we now have $p_a=\Pr[a_n>0]\times\Pr[h_n=0]=\alpha^2$ and similarly $p_h=(1-\alpha)^2$ and $p_0 = 2\alpha(1-\alpha)$.
We plot the probabilities of winning the independent SSLE private game and compare them to the SSLE grinding game in Figure~\ref{fig:grinding-ssle-10}, \ref{fig:grinding-ssle-20} and~\ref{fig:grinding-ssle-30}.
We indeed notice that the independent SSLE game performs much better than the grinding game but, surprisingly, also better than the SSLE game.
The difference between the independent SSLE and SSLE private games can be explained by the fact that the adversary in the independent game is much less likely to be elected sole leader ($\alpha^2$ vs $\alpha$) and, hence, the gap in this case increases less often. The difference between the independent SSLE and PLE is also explained similarly: the probability of the gap increase goes from $e^{1-\alpha}-e^{-1}$ to $(1-\alpha)^2$.
Although the persistence parameter is smaller in this case, it is also expected that there will be rounds with no winner as well as more natural forks (unlike in the SSLE case).
%Although the private attack is less successful in this case
% We see that even in the SSLE case, where the security threshold is quite high, the persistence parameter is orders of magnitude more than in the previous section.
% Perhaps unsurprisingly, we also notice that the probability in the grinding games are at least an order of magnitude more than in the simple private attack.
% For example for $\alpha = 0.1$,
% the probability of winning the SSLE game of length 5 is 0.002 meaning that the $2.10^{-3}-$persistence parameter is at least 5, whereas in the simple private attack with a length of $5$ gives a probability of winning the game for any length greater than 5 of 0.02.
% This is less true in the PLE case where the probabilities are of the same order
% (e.g. $\sim 0.07$ for $\alpha = 0.1$ and $n=5$).
% % For example for $\alpha = ..$ and a length of 9000, gives a probability of $\simeq 10^{-4}$ 
% % whereas for SSLE, a parameter of 450 gave a security of $10^{-13}$.
% % Similarly with a $0.198$ adversary, a persistence parameter of $1000$ gives a security of $10^{-3}$.
% %$\text{VRF}_{sk}(r_{i-1},\text{time})<\mathsf{target}\times \text{stake}$

\section{Conclusion and Future work}
In this work, we have performed a comparison of private attacks against longest-chain PoS protocols that use SSLE and PLE. 
We have found that the persistence parameter under this specific attack is reduced significantly when using SSLE, by around 25\% against a 25 or 33\% adversary.
We also found that the security threshold against 
grinding private attacks is higher ($\simeq 0.36$) in the SSLE case than in the PLE case ($\simeq 0.26$).
% The motivation for focusing on private attacks comes from the work of Dembo et al.~\cite{dembo2020everything} that show that the true security threshold is the same as the private attack security threshold in longest-chain protocol.
These results are encouraging and should help convince real-world system designers to make the switch to SSLE in their PoS blockchains.
For future work, it will be interesting to compare the results based on other attacks (e.g., balance attack~\cite{nakamoto-pos} or general case), as well as to relax assumptions such as the synchronous network or static adversary. 
% It is possible that the difference between both will be quite different (the difference between the two is even worse with balance attack since with SSLE the adversary is elected leader more often and hence can do the attack (i.e. maintain both chain) more often.

\section*{Acknowledgments}
The authors would like to thank Jorge Soares for his valuable feedback on this paper. D.C. was supported by the MIUR grant `Dipartimenti di Eccellenza 2018-2022' (E11G18000350001).

% ---- Bibliography ----
\bibliographystyle{plain}  % we recommend the plain bibliography style
\bibliography{ref}       % xampl.bib comes with BibTeX
\end{document}